\documentclass{article}
\usepackage{amsfonts}
\usepackage{amsmath}
\usepackage{ hyperref, mathtools, amssymb, mathrsfs, amsthm, float, caption, array, multirow, tabularx, multicol}
\usepackage[x11names]{xcolor}

\usepackage[left=4cm, right=4cm, top=4cm,  bottom=2cm]{geometry}

{\title{Some remarks on the bias distribution analysis of discrete-time identification algorithms based on pseudo-linear regressions }}

\author  {Bernard Vau\thanks {  bernard.vau@satie.ens-cachan.fr}, Henri Bourl{\`e}s \\ \small{ SATIE, Ecole normale sup{\'e}rieure Paris-Saclay  94230 Cachan France} }

\begin{document}

\maketitle

\newtheorem{proposition}{Proposition}
\newtheorem{lemma}{Lemma}
\newtheorem{theorem}{Theorem}

\begin{abstract}
In 1998, A. Karimi and I.D. Landau published in the journal "Systems and Control letters" an article
entitled \textquotedblleft Comparison of the closed-loop identification
methods in terms of bias distribution\textquotedblright . One of its main
purposes was to provide a bias distribution analysis in the frequency domain
of closed-loop output error identification algorithms that had been recently
developed. The expressions provided in that paper are only
valid for prediction error identification methods (PEM), not for pseudo-linear
regression (PLR) ones, for which we give the correct frequency domain bias
analysis, both in open- and closed-loop. Although PLR was initially (and is
still) considered as an approximation of PEM, we show that it gives better
results at high frequencies.
 
\end{abstract}



\section{Introduction}

In the field of discrete-time identification, the bias distribution analysis
over frequency domain is very powerful in order to evaluate the influence of
input and noise spectra on an identified model, and to assess qualitatively
the model that can be obtained if it has not the same structure as the
identified system. This method was first introduced in \cite{r1} for
open-loop identification; see \cite{r2} for more details, especially in the
context of closed-loop operations. In these two references, this bias
analysis has been developed in the perspective of prediction error methods
(PEM), which aim at minimizing a one step further prediction
error variance. In \cite{r3} and \cite{r4}, Karimi and Landau used the same
method to infer the bias distribution of closed-loop algorithms. In \cite{r3} (section 5) and \cite{r15} (section 6), it is claimed that this analysis is valid for CLOE, F-CLOE and AF-CLOE algorithms which are of PLR type. However, this is not true as already noticed in \cite{r4}, p.308, although PLR was
initially (and is still) considered as an approximation of PEM (see, e.g.,
remark 8.4.2 of \cite{r14}); rather, PLR algorithms tend to cancel the correlation
function between the prediction error and an observation vector, which in
general is the regressor of the predictor, possibly filtered (see the
correlation approach developed by Ljung in \cite{r2}). If the system dynamics
is approximately known beforehand, the difference between PEM and PLR can be
made quite small, as shown in section 4, with an appropriate regressor
filtering. But this approach presupposes that what one is seeking is already
known, a vicious circle which prevents from addressing the core of the problem.
It is said in \cite{r4} (p. 308) that the bias distribution cannot be computed for CLOE algorithm.
Therefore, bias distribution of PLR algorithms, such as CLOE, is an open problem, which is solved here; this bias distribution is determined and we show that it is quite different from that of PEM. To do this we
introduce in section 3, the concept of equivalent prediction error -- most
of the time a non-measurable signal -- whose variance is effectively
minimized by the PLR algorithm, even if the identified system is not in the
model set. This approach shows that compared to PEM, PLR schemes strongly
penalize the model misfit at high frequency, in a way comparable to the
classical open-loop least-squares algorithm, whatever the predictor model is
(output error, ARMAX, etc.). In section 5, an example is given, that relies
on the Vinnicombe gap, in order to compare the model misfit of a PEM scheme
and its corresponding PLR one: It brings to light the discrepancy between
the two methods in case of a closed-loop output error (CLOE) identification
structure.

\section{Optimal prediction error and bias distribution analysis of PEM
algorithms}

At first, let us recall briefly the model structures used here, both in
open-and closed-loop, and the manner to obtain the bias distribution from
PEM algorithms. According to Landau et al. (\cite{r3}, p. 44), we
distinguish between equation error models and output error models. Equation
error model in open-loop corresponds to the equation: 
\begin{equation*}
A(q^{-1})y(t)=B(q^{-1})u(t)+C(q^{-1})e(t)
\end{equation*}%
with $G(q^{-1})=\frac{B(q^{-1})}{A(q^{-1})}$, and $W(q^{-1})=\frac{C(q^{-1})%
}{A(q^{-1})}$, respectively the deterministic and stochastic parts of the
model; $q^{-1}$ the shift backward operator, $A(q^{-1})$ is a monic
polynomial, $W(q^{-1})$ is the ratio of two monic polynomials, and $%
u(t),y(t),e(t)$ are the input, the output, and a centered gaussian white
noise, respectively.

\noindent According to the noise structure, we distinguish between the following cases:

\begin{itemize}
\item when $C(q^{-1})=1,$ corresponding to the ARX model,

\item when $C(q^{-1})$ is a monic polynomial of degree strictly greater than
0, corresponding to the ARMAX model.
\end{itemize}

\noindent Other equation error models exist (e.g. ARARMAX, etc.) but they are not
treated in this paper. On the other hand, the output error model in
open-loop is given by%
\begin{equation*}
y(t)=G(q^{-1})u(t)+v(t)
\end{equation*}%
where $v(t)$ is a centered gaussian noise not necessarily white, but
uncorrelated with the input.

Let us call $\widehat{G}(q^{-1})$ and $\widehat{W}(q^{-1})$ the estimations
of $G(q^{-1})$ and $W(q^{-1}),$ respectively. In the case of an open-loop
equation error model, the prediction error is given by (\cite{r4}, (3.3)), (%
\cite{r4}, (9.62): 
\begin{equation}
\varepsilon (t)=\widehat{W}(q^{-1})\{(G(q^{-1})-\widehat{G}%
(q^{-1}))u(t)+(W(q^{-1})-\widehat{W}(q^{-1}))e(t)\}+e(t)  \label{e1}
\end{equation}%
whereas the optimal error for the open-loop output error model is simply: 
\begin{equation*}
\varepsilon (t)=(G(q^{-1})-\widehat{G}(q^{-1}))u(t)+v(t)
\end{equation*}%

\noindent The closed-loop case is more complicated, due to the feedback control law.
We assume that the controller has an R-S structure, i.e. $%
S(q^{-1})u(t)=-R(q^{-1})y(t);$ let us define the direct sensitivity function
(transfer function from the output noise to the output): 
\begin{equation*}
S_{yp}(q^{-1})=\frac{A(q^{-1})S(q^{-1})}{%
A(q^{-1})S(q^{-1})+B(q^{-1})R(q^{-1})}
\end{equation*}%

\noindent In the context of an equation error model, in which the model input is given
by $S(q^{-1})\widehat{u}(t)=-R(q^{-1})\widehat{y}(t),$ the optimal predicted
output is 
\begin{equation*}
\widehat{y}(t)=\widehat{G}\widehat{u}(t)+\widehat{\overline{W}}%
(q^{-1})\varepsilon (t)
\end{equation*}%
where $\widehat{\overline{W}}(q^{-1})=\widehat{W}(q^{-1})-\widehat{S}%
_{yp}^{-1}(q^{-1})$ (see \cite{r3}, eq. (5.7) sq.), thus we have: 
\begin{equation*}
\varepsilon (t)=\widehat{W}^{-1}(q^{-1})\{(G-\widehat{G}%
)S_{yp}r_{u}(t)+(WS_{yp}{\widehat{S}}_{yp}^{-1}-\widehat{W})e(t)\}+e(t)
\end{equation*}%
\noindent This expression is directly obtained from (\cite{r3}, (5.12)).

\noindent The optimal predicted output is
given by $\widehat{y}(t)=\widehat{G}\widehat{u}(t),$ and the corresponding
optimal prediction error of output error model by 
\begin{equation*}
\varepsilon (t)=\widehat{S}_{yp}(G-\widehat{G})S_{yp}r_{u}(t)+S_{yp}v(t)
\end{equation*}%
where $\widehat{S}_{yp}=\frac{\widehat{A}S}{\widehat{A}S+\widehat{B}R},$ $%
\widehat{A},\widehat{B}$ being the estimations of $A$ and $B$ respectively
(where the dependance in $q^{-1}$ is omitted).

\noindent The purpose of PEM algorithms is to minimize the prediction error variance $%
\mathbf{E}[\varepsilon ^{2}(t)]$, and whatever the algorithm structure is,
both in open or closed-loop, one obtains the optimal estimated parameter
vector $\widehat{\theta }_{PEM}^{\ast }$: 
\begin{equation}
\widehat{\theta }_{PEM}^{\ast }=Arg\min \mathbf{E}[\varepsilon
^{2}(t)]=Arg\min {\int_{-\pi }^{+\pi }}\left\vert\mathcal{Z}\{ \varepsilon\} (e^{i\omega
})\right\vert ^{2}d\omega  \label{e5}
\end{equation}%
\noindent where $\mathcal{Z}$ is the z-transform. This expression is at the origin of bias analysis for PEM algorithms.

\section{Bias distribution of pseudo-linear regression algorithms\label%
{sect_bias_PLR}}

\subsection{Equivalent prediction error}

The \emph{a posteriori} model predicted output $\widehat{y}(t+1)$ is
provided by $\widehat{y}(t+1)=\widehat{\theta }(t+1)\phi (t,\widehat{\theta }%
)$, where $\phi (t,\widehat{\theta })$ is the regressor structure, generally
depending on $\widehat{\theta }$. The \emph{a posteriori} prediction error
is given by the expression: $\varepsilon (t+1)=y(t+1)-\widehat{y}(t+1,%
\widehat{\theta })$. Most of the PLR identification procedures are solved
recursively with the so-called parameter adaptation algorithm (PAA): 
\begin{subequations}
\begin{align}
\widehat{\theta }(t+1)& =\widehat{\theta }(t)+F(t)\phi (t)\varepsilon (t+1)
\\
F^{-1}(t+1)& =\lambda _{1}F^{-1}(t)+\lambda _{2}\phi (t)\phi ^{T}(t)
\end{align}%
$F(t)$ is the adaptation gain (positive definite matrix), $0<\lambda
_{1}\leq 1,0\leq \lambda _{2}<2$ are forgetting factors. The stationary
condition of the PAA is (see \cite{r2}, p. 224): 
\end{subequations}
\begin{equation}
\mathbf{E}[\varepsilon (t+1)\phi (t)]=0  \label{station_cond}
\end{equation}

\begin{lemma}
In general, the stationarity condition of the parameter adaptation algorithm 
$\mathbf{E} [\varepsilon (t +1) \phi (t ,\widehat{\theta })] =0$, is not the
one associated with the prediction error variance $\mathbf{E} [\varepsilon
^{2} (t)]$ minimization.

\label{prop1}

\begin{proof}
As a counterexample let us consider the extended least squares algorithm
(corresponding to an ARMAX model), for which the predicted output is $%
\widehat{y}(t+1)=\widehat{\theta }^{T}\phi (t)$. Let $\varepsilon (t)=y(t)-%
\widehat{y}(t),$ $\widehat{\theta }^{T}=[%
\begin{array}{cccccc}
\widehat{a}_{1}, & \cdots  & \widehat{b}_{1}, & \cdots  & \widehat{c}_{1}, & 
\cdots 
\end{array}%
]$ and $\phi (t)^{T}=[%
\begin{array}{cccccc}
-y(t), & \cdots  & u(t), & \cdots  & \varepsilon (t), & \cdots 
\end{array}%
]$. Let $\widehat{C}(q^{-1})=1+\widehat{c}_{1}q^{-1}+\widehat{c}%
_{2}q^{-2}\cdots $ and assume that $\widehat{C}(q^{-1})\neq 1.$ One has that 
$\phi (t,\widehat{\theta })=-\widehat{C}(q^{-1},\widehat{\theta })\frac{%
\partial \varepsilon (t+1)}{\partial \widehat{\theta }}\text{.}$ In this
case the stationary condition of the parameter adaptation algorithm is $%
\mathbf{E}[\varepsilon (t+1)\widehat{C}(q^{-1})\frac{\partial \varepsilon
(t+1)}{\partial \widehat{\theta }}]=0$. Therefore $\frac{\partial }{\partial 
\widehat{\theta }}\mathbf{E}\left[ \varepsilon \left( t+1\right) ^{2}\right]
=2\mathbf{E}\left[ \varepsilon \left( t+1\right) \frac{\partial \varepsilon
\left( t+1\right) }{\partial \widehat{\theta }}\right] $ cannot be zero
unless $\left\{ \varepsilon \left( t+1\right) \right\} $ is white, and that
cannot happen if the system is not in the model set.
\end{proof}
\end{lemma}

\begin{lemma}
In PLR schemes, the stationarity condition $\mathbf{E}\left[ \varepsilon
\left( t+1\right)\phi \left( t,\hat{\theta} \right) \right] =0$ is the
stationary condition of the variance minimization problem of the signal $%
\varepsilon _{E}(t+1,\widehat{\theta }),$ called the "equivalent prediction
error" in the sequel, and characterized by the following two conditions:

1) If the system is in the model set and if the estimated parameter vector $%
\widehat{\theta }^{*}_{PLR}$ is equal to the true parameters vector $\theta $
$(\theta =\widehat{\theta }^{*}_{PLR})\text{,}$ then one has:

\begin{itemize}
\item {for the equation error model: \newline
$\varepsilon _{E}(t+1,\widehat{\theta }_{PLR}^{\ast })=\varepsilon
_{E}(t+1,\theta )=\varepsilon (t+1,\theta )=e(t+1)$}

\item {for the open-loop output error model:\newline
$\varepsilon _{E}(t+1,\widehat{\theta }_{PLR}^{\ast })=\varepsilon
_{E}(t+1,\theta )=\varepsilon (t+1,\theta )=v(t+1)$}

\item {for the closed-loop output error model: \newline
$\varepsilon _{E}(t+1,\widehat{\theta }_{PLR}^{\ast })=\varepsilon
_{E}(t+1,\theta )=\varepsilon (t+1,\theta )=S_{yp}v(t+1)$ }
\end{itemize}

2) The vector $\phi _{E}(t)=-\frac{\partial \varepsilon _{E}(t+1)}{\partial 
\widehat{\theta }},$ called the "equivalent regressor", is not a function of 
$\widehat{\theta }\text{,}$ i.e $\frac{\partial \phi _{E}(t)}{\partial 
\widehat{\theta }}=0.$
\label{prop2}
\begin{proof}
By 2),$\frac{\partial \phi _{E}(t)}{\partial \widehat{\theta }}=0\text{,}$
so that $\varepsilon _{E}(t+1)=-\widehat{\theta }^{T}\phi _{E}(t)+k$ ($k$
independant of $\widehat{\theta }$). By 1) we get:\newline
$e(t+1)=-\theta ^{T}\phi _{E}(t)+k$ for the equation error model, \newline
$v(t+1)=-\theta ^{T}\phi _{E}(t)+k$ for the open-loop output error model,%
\newline
$S_{yp}v(t+1)=-\theta ^{T}\phi _{E}(t)+k$ for the closed-loop output error
model.\newline
Combining the preceding equations, one gets the following prediction error
expressions:\newline
$\varepsilon _{E}(t+1)=(\theta -\widehat{\theta })^{T}\phi _{E}(t)+e(t+1)$
for the equation error model, \newline
$\varepsilon _{E}(t+1)=(\theta -\widehat{\theta })^{T}\phi _{E}(t)+v(t+1)$
for the open-loop output error model, \newline
$\varepsilon _{E}(t+1)=(\theta -\widehat{\theta })^{T}\phi
_{E}(t)+S_{yp}v(t+1)$ for the closed-loop output error model. \newline
The stationarity condition applied to these equations is $\mathbf{E}\left[
\varepsilon _{E}\left( t+1\right) \phi _{E}\left( t\right) \right] =0.$
Since $\phi _{E}(t)=-\frac{\partial \varepsilon _{E}(t+1)}{\partial \widehat{%
\theta }},$ this stationnarity condition can be rewritten $\mathbf{E}%
[\varepsilon _{E}(t+1)\frac{\partial \varepsilon _{E}}{\partial \widehat{%
\theta }}]=0\text{,}$ which is the gradient of $\mathbf{E}[\varepsilon
_{E}^{2}(t+1)]$ with respect to $\widehat{\theta }$.
\end{proof}
\end{lemma}

\begin{lemma}
One has $\widehat{\theta }_{PLR}^{\ast }=Arg\min {\int_{-\pi }^{\pi }}%
|\mathcal{Z}\{\varepsilon _{E}\}(e^{i\omega })\|^{2}d\omega \text{.}$

\begin{proof}
The minimization problem of $\mathbf{E} [\varepsilon _{E}^{2} (t)]$ is
convex since $\varepsilon_{E}\left(t\right)$ is expressed linearly in
function of $\phi_{E}\left(t\right)$, as shown by the proof of lemma 2.
\end{proof}
\end{lemma}

In this paper, we assume that the dependence of $\widehat{y}(t+1)$ with
respect to $\widehat{\theta }$ can be expressed uniquely \textit{via} $\varepsilon (t,%
\widehat{\theta })$, which is the case of the ARMAX and output error
predictors (in open- and closed-loop operations), as shown in Table \ref%
{table1}. Under this assumption, we can write: $\frac{\partial \varepsilon
(t+1)}{\partial \widehat{\theta }}=\frac{-\partial \widehat{y}(t+1)}{%
\partial \widehat{\theta }}=-\phi (t)-\widehat{\theta }^{T}\frac{\partial
\phi }{\partial (q\varepsilon )}\frac{\partial \varepsilon (t+1)}{\partial 
\widehat{\theta }}\text{.}$ Let us define $Q(q^{-1},\widehat{\theta })=1+%
\widehat{\theta }^{T}\frac{\partial \phi }{\partial (q\varepsilon )},$ so
that $Q(q^{-1},\widehat{\theta })\frac{\partial \varepsilon (t+1)}{\partial 
\widehat{\theta }}=-\phi (t)\text{.}$ The computation of $\frac{\partial
\phi }{\partial (q\varepsilon )}$ can be performed using the remarks in
 \ref{appendix_A}. We can now state the main result of this section.

\begin{theorem}
\label{theorem1}

Let $Q (q^{ -1} ,\widehat{\theta }) =1 +\widehat{\theta }^{T} \frac{ \partial \phi}{ \partial (q \varepsilon)}$. Assuming that  $\frac{ \partial ^{2}\phi}{ \partial (q \varepsilon) \partial \widehat{\theta }} =0$,
the equivalent prediction error defined in lemma \ref{prop2} is such that: 

\begin{itemize}
 \item for the equation error model: \\
 $ \varepsilon _{E} (t +1) =Q (q^{ -1} ,\widehat{\theta }) \varepsilon  (t +1)+\left(1 -Q (q^{ -1} ,\widehat{\theta })\right)e(t +1)$ 
  
\item for the open-loop output error model:\\
$\varepsilon _{E} (t +1) =Q (q^{ -1} ,\widehat{\theta }) \varepsilon  (t +1) +\left(1-Q (q^{ -1} ,\widehat{\theta })\right)v(t +1)$ 

\item for the closed-loop output error model: \\
$\varepsilon _{E} (t +1) =Q (q^{ -1} ,\widehat{\theta }) \varepsilon  (t +1)+\left(1 -Q (q^{ -1} ,\widehat{\theta }) \right)S_{yp}(q^{-1})v (t +1)$ 
 
\end{itemize}

\label{prop4}

\end{theorem}

\begin{proof}
One has to verify Conditions 1) and 2) of Lemma 2.
1): If the system is in the model
set and $\theta  =\widehat{\theta }$, it is immediate to check that: \\
 $\varepsilon _{E} (t +1) =e (t +1)$ for the equation error model,\\
 $\varepsilon _{E} (t +1) =v (t +1)$ for open-loop output error model,\\
$\varepsilon _{E} (t +1) =S_{yp}(q^{-1})v (t +1)$ for closed-loop output error model.\\
2): One has to verify that $\phi _{E} (t) =\frac{ - \partial \varepsilon _{E} (t +1)}{ \partial \widehat{\theta }}$
 is not a function of $\widehat{\theta }$,  i.e. $\frac{ \partial \phi _{E} (t)}{ \partial \widehat{\theta }} =0$. 
For this purpose let us compute: \\
$\frac{ \partial \varepsilon _{E} (t +1)}{ \partial \widehat{\theta }} =Q \frac{ \partial \varepsilon  (t +1)}{ \partial \widehat{\theta }} +\frac{ \partial Q}{ \partial \widehat{\theta }} (\varepsilon  (t +1) -e (t +1)) = -\phi  (t) +\frac{ \partial Q}{ \partial \widehat{\theta }} (\varepsilon  (t +1) -e (t +1))$
 for the equation error model.  But $\frac{ \partial Q}{ \partial \widehat{\theta }} =\frac{ \partial \phi}{ \partial(q \varepsilon)} +\widehat{\theta }^{T} \frac{ \partial ^{2}\phi}{ \partial (q\varepsilon)  \partial \widehat{\theta }}\text{.}$ 
Under the assumption $\frac{ \partial ^{2}\phi}{ \partial (q\varepsilon)  \partial \widehat{\theta }} =0\text{,}$
 we obtain $\frac{ \partial Q}{ \partial \widehat{\theta }} =\frac{ \partial \phi}{ \partial(q \varepsilon)}$. 
As $\phi _{E} (t) =\frac{ - \partial \varepsilon _{E} (t +1)}{ \partial \widehat{\theta }} ,$ one gets $\frac{ \partial \phi _{E} (t)}{ \partial \widehat{\theta }} = -\frac{ \partial \phi  (t)}{ \partial \widehat{\theta }} +\frac{ \partial \phi  (t)}{ \partial \widehat{\theta }} =0.$ 
The same result holds for output error models (both in open- and closed-loop), substituting $e(t+1)$ by $v(t+1)$ and by $S_{yp}(q^{-1})v(t+1)$ respectively.

\end{proof}

\subsection{Expression of the bias distribution of PLR algorithms}

For each noise model, either in open- or closed-loop operations, the
corresponding PLR algorithm has a specific predictor. We give in Table \ref%
{table1} these expressions.



\begin{center}
\begin{table}[h]
\resizebox{1\textwidth}{!}{\begin{minipage}{\textwidth}

\renewcommand*{\arraystretch}{1}

\begin{tabular}{|>{\centering\scriptsize\arraybackslash}p{5pt}|>{\centering\scriptsize\arraybackslash}p{40pt}|>{\centering\scriptsize\arraybackslash}p{285pt}|}

		\hline

		\multicolumn{2}{|>{\centering\scriptsize}p{45pt}|} {NOISE MODEL} &  {PREDICTED OUTPUT $\widehat{y}(t+1)$ and REGRESSOR $\phi(t)$} \\[5pt]

			\hline
		   \multirow{8}*{\rotatebox[origin=c]{90}{\centerline{  Open-loop}}} 
			&  \multirow{2}*{{ARX}} &
			{$\widehat{y} (t +1) = -\widehat{a}_{1} y (t) -\widehat{a}_{2} y (t -1) +\cdots  +\widehat{b}_{1} u (t) +\widehat{b}_{2} u (t -1) +		 \cdots $}\\  \cline{3-3}
			&  & {$\phi(t)=[-y(t), -y(t-1), \cdots u(t), u(t-1), \cdots ]$}\\ \cline{2-3}
			
			  & \multirow{2}*{{ARMAX}} &
				{ $\widehat{y} (t +1) = -\widehat{a}_{1} y(t) -\widehat{a}_{2} y(t -1) +\cdots  +\widehat{b}_{1} u (t) +\widehat{b}_{2} u (t -1) +\cdots                                         +   	\widehat{c}_{1} \varepsilon  (t) +\widehat{c}_{2} \varepsilon  (t -1) +\cdots $ } \\ \cline{3-3}
		  &  &{$\phi(t)=[-y(t), -y(t-1), \cdots u(t), u(t-1), \cdots  \varepsilon(t), \varepsilon(t-1), \cdots ]$}\\   \cline{2-3}
		
			  & OUTPUT &
				{ $\widehat{y} (t +1) = -\widehat{a}_{1} \widehat{y} (t) -\widehat{a}_{2} \widehat{y} (t -1) +\cdots  +\widehat{b}_{1} u (t) +\widehat{b}_{2} u (t -1) + \cdots $}\\ \cline{3-3}
				 & ERROR & {$\phi(t)=[-\widehat{y}(t), -\widehat{y}(t-1), \cdots u(t), u(t-1), \cdots ]$}\\ \cline{2-3}
			\hline

			\multirow{6}*{\rotatebox[origin=l]{90}{\centerline{   Closed-loop}}} 
			
			 & \multirow{2}*{{ARMAX}} & { $\widehat{y} (t +1) = -\widehat{a}_{1} \widehat{y} (t) -\widehat{a}_{2} \widehat{y} (t -1) +\cdots  +\widehat{b}_{1} \widehat{u} (t) +\widehat{b}_{2} \widehat{u} (t -1) +\cdots  +\widehat{h}_{1} \varepsilon  (t) +\widehat{h}_{2} \varepsilon  (t -1) + \cdots $ }\\
			\cline{3-3}
			& &  {$\phi(t)=[-\widehat{y}(t), -\widehat{y}(t-1), \cdots \widehat{u}(t), \widehat{u}(t-1), \cdots  \varepsilon(t), \varepsilon(t-1), \cdots ]$}\\   \cline{2-3}
			& OUTPUT & {  $\widehat{y} (t +1) = -\widehat{a}_{1} \widehat{y} (t) -\widehat{a}_{2} \widehat{y} (t -1) + \cdots  +    		    					\widehat{b}_{1} \widehat{u} (t) +\widehat{b}_{2} \widehat{u }(t -1) + \cdots $} \\ \cline{3-3}
			 & ERROR & {$\phi(t)=[-\widehat{y}(t), -\widehat{y}(t-1), \cdots \widehat{u}(t), \widehat{u}(t-1), \cdots ]$}\\ \cline{2-3}
			
			\hline

\end{tabular}

\caption{Predicted output in open- and closed-loop operations}
\label{table1}

\end{minipage}}
\end{table}
\end{center}

\noindent (Table \ref{table2} provides the expressions of $Q(q^{-1},\widehat{\theta })$
and $\varepsilon _{E}(t+1)$ for all predictor models presented above where $%
\widehat{P}=\widehat{A}S+\widehat{B}R$).

\begin{center}
\begin{table}[h]
\resizebox{1\textwidth}{!}{\begin{minipage}{\textwidth}
\renewcommand*{\arraystretch}{1}

\begin{tabular}{|>{\centering\scriptsize\arraybackslash}p{40pt}|>{\centering\scriptsize\arraybackslash}p{40pt}|>{\centering\scriptsize\arraybackslash}p{40pt}|>{\centering\scriptsize\arraybackslash}p{200pt}|}

\hline

 \multicolumn{2}  {|>{\centering\scriptsize}p{45pt}|}{NOISE MODEL} & {\scriptsize ${Q} (\widehat{\theta })$} & {\scriptsize $\varepsilon _{E} (t)$}  \\
\hline
 
\multirow{4}*{\rotatebox[origin=c]{0}{\centerline{Open-\ loop}}} 
&  ARX &  $1$ & {$\widehat{A} (G -\widehat{G}) u (t) +\widehat{A} (W -\frac{1}{\widehat{A}}) e (t) +e (t)$}\\ \cline{2-4}

&  ARMAX &  $\widehat{C}(q^{-1},\widehat\theta)$ &  {$\widehat{A} (G -\widehat{G}) u (t) +\widehat{A} (W -\frac{\widehat{C}}{\widehat{A}}) e (t) +e (t)$}\\ \cline{2-4}
&  OUTPUT ERROR  & $\widehat{A}(q^{-1},\widehat\theta)$& $\widehat{A} (G -\widehat{G}) u (t) +v (t)$\\  \hline

\multirow{3}*{\rotatebox[origin=c]{0}{\centerline{Closed-\ loop  }}} 
& ARMAX &  $\widehat{C}(q^{-1},\widehat\theta) $ & $\widehat{A} \ (G -\widehat{G}) S_{yp}r_{u} (t) +\widehat{A} \left( WS_{yp}\widehat{S}_{yp}^{-1}-\frac{\widehat{C}}{\widehat{A}}\right)e(t) +e (t)$ \\ \cline{2-4}

 &  OUTPUT ERROR  & $\frac{\widehat{P}(q^{-1},\widehat\theta)}{S(q^{-1})}$ & $\widehat{A} (G -\widehat{G})S_{yp}(q^{-1})r_{u}(t) +S_{yp}(q^{-1})v (t)$\\  \hline

\end{tabular}

\caption{Expression of the equivalent prediction error (open-loop and closed-loop)}

\label{table2}

\end{minipage}}
\end{table}
\end{center}

\noindent These expressions differ from those of PEM that can be computed from (\ref{e5}). In particular one can notice that for ARMAX and OE models, both in open- and closed-loop, the deterministic part of $\varepsilon _{E}(t)$ is weighted by $\widehat{A}$ which is the same weighting function as for the ARX model. This function strongly penalizes the high frequency misfit. Table \ref{table3} gives the expressions of the asymptotic values of the estimated parameters ${\widehat{\theta }}_{PLR}^{\ast }$ in the frequency domain. 
The expressions ${{\Phi}}_{uu}(\omega)$, ${{\Phi}}_{r_{u}r_u}(\omega)$, ${{\Phi}}_{ee}(\omega)$, correspond to the spectra of $u$, $r_{u}$ (the additive excitation on the input in closed-loop), and $e$ respectively.

\begin{center}
\begin{table}[h]
\resizebox{1\textwidth}{!}{\begin{minipage}{\textwidth}
\renewcommand*{\arraystretch}{1}
\begin{tabular}{|>{\centering\scriptsize\arraybackslash}p{5pt}|>{\centering\scriptsize\arraybackslash}p{70pt}|>{\centering\scriptsize\arraybackslash}p{260pt}|}

\hline

 \multicolumn{2}  {|>{\centering\scriptsize}p{75pt}|}{NOISE MODEL} &  {${\widehat{\theta}}_{PLR}^*$}  \\
\hline
 
\multirow{4}*{\rotatebox[origin=c]{90}{\rightline{Open-loop}}} 
&  ARX & { $ Argmin \displaystyle\int\nolimits_{-\pi}^{+\pi}  \left\{\left| \widehat{A} (e^{i\omega}) \right |^2\left(   \left|G(e^{i\omega})- \widehat{G} (e^{i\omega}) \right|^2 {\Phi}_{uu} (\omega) +\right. \right. \newline \left.\left. \cdots \left|W(e^{i\omega})-\frac{1}{\widehat{A}(e^{i\omega})} \right|^2 {\Phi}_{ee}(\omega) \right)\right\} \mathrm d\omega $ }\\ \cline{2-3}

&  ARMAX & { $ Argmin \displaystyle\int\nolimits_{-\pi}^{+\pi}  \left\{\left| \widehat{A} (e^{i\omega}) \right |^2 \left( \left|G(e^{i\omega})- \widehat{G} (e^{i\omega}) \right|^2 {\Phi}_{uu} (\omega)+\right.\right. \newline \cdots  \left.\left. \left|W(e^{i\omega})-\frac{\widehat{C}}{\widehat{A}(e^{i\omega})} \right|^2 {\Phi}_{ee}(\omega)\right)\right\} \mathrm d\omega $ }\\ \cline{2-3}

&  OUTPUT ERROR & { $ Argmin \displaystyle\int\nolimits_{-\pi}^{+\pi}  \left| \widehat{A} (e^{i\omega}) \right |^2  \left|G(e^{i\omega})- \widehat{G} (e^{i\omega}) \right|^2 {\Phi}_{uu} (\omega)  \mathrm d\omega $ }\\ 
\hline

\multirow{3}*{\rotatebox[origin=c]{90}{\rightline{Closed-loop}}} 
&  ARMAX &
 { $ Argmin \displaystyle\int\nolimits_{-\pi}^{+\pi}  \left\{\left| \widehat{A} (e^{i\omega}) \right |^2 \left( \left|G(e^{i\omega})- \widehat{G} (e^{i\omega}) \right|^2 \left|S_{yp}(e^{i\omega}) \right|^2 {\Phi}_{r_ur_u} (\omega) +\right.\right. \newline  \left.\left. \cdots  \left|W(e^{i\omega})\frac{S_{yp}(e^{i\omega})}{\widehat{S}_{yp}^{-1}(e^{i\omega})}-\frac{\widehat{C}(e^{i\omega})}{\widehat{A}(e^{i\omega})} \right|^2  {\Phi}_{ee}(\omega) \right)\right\} \mathrm d\omega $ }\\ \cline{2-3}

& OUTPUT ERROR &{ $ Argmin \displaystyle\int\nolimits_{-\pi}^{+\pi}  \left| \widehat{A} (e^{i\omega})\right |^2  \left|G(e^{i\omega})- \widehat{G} (e^{i\omega}) \right|^2 \left|S_{yp}(e^{i\omega}) \right|^2 {\Phi}_{r_ur_u} (\omega) \mathrm{d} \omega$}\\

\hline

\end{tabular}

\caption{Bias distribution for PLR algorithms in open- and closed-loop}

\label{table3}

\end{minipage}}
\end{table}
\end{center}

\noindent By comparison, we recall in Table \ref{table3b} the expressions for the estimated parameters in case of PEM methods.
The expressions corresponding to ARMAX and output error models  result directly from (4.186) and (4.193) of \cite{r4} in an open-loop context, and from (9.81) and (9.79) of the same book in closed-loop operation.

\begin{center}
\begin{table}[h]
\resizebox{1\textwidth}{!}{\begin{minipage}{\textwidth}
\renewcommand*{\arraystretch}{1}
\begin{tabular}{|>{\centering\scriptsize\arraybackslash}p{5pt}|>{\centering\scriptsize\arraybackslash}p{70pt}|>{\centering\scriptsize\arraybackslash}p{260pt}|}

\hline

 \multicolumn{2}  {|>{\centering\scriptsize}p{75pt}|}{NOISE MODEL} &  {${\widehat{\theta}}_{PEM}^*$}  \\
\hline
 
\multirow{4}*{\rotatebox[origin=c]{90}{\rightline{Open-loop}}} 
&  ARX & { $ Argmin \displaystyle\int\nolimits_{-\pi}^{+\pi}  \left\{\left| \widehat{A} (e^{i\omega}) \right |^2\left(   \left|G(e^{i\omega})- \widehat{G} (e^{i\omega}) \right|^2 {\Phi}_{uu} (\omega) +\right. \right. \newline \left.\left. \cdots \left|W(e^{i\omega})-\frac{1}{\widehat{A}(e^{i\omega})} \right|^2 {\Phi}_{ee}(\omega) \right)\right\} \mathrm d\omega $ }\\ \cline{2-3}

&  ARMAX & { $ Argmin \displaystyle\int\nolimits_{-\pi}^{+\pi}  \left\{\left| \frac{\widehat{A} (e^{i\omega})}{\widehat{C} (e^{i\omega})} \right |^2 \left( \left|G(e^{i\omega})- \widehat{G} (e^{i\omega}) \right|^2 {\Phi}_{uu} (\omega)+\right.\right. \newline \cdots  \left.\left. \left|W(e^{i\omega})-\frac{\widehat{C}}{\widehat{A}(e^{i\omega})} \right|^2 {\Phi}_{ee}(\omega)\right)\right\} \mathrm d\omega $ }\\ \cline{2-3}

&  OUTPUT ERROR & { $ Argmin \displaystyle\int\nolimits_{-\pi}^{+\pi}    \left|G(e^{i\omega})- \widehat{G} (e^{i\omega}) \right|^2 {\Phi}_{uu} (\omega)  \mathrm d\omega $ }\\ 
\hline

\multirow{3}*{\rotatebox[origin=c]{90}{\rightline{Closed-loop}}} 
&  ARMAX &
 { $ Argmin \displaystyle\int\nolimits_{-\pi}^{+\pi}  \left\{\left| \frac{\widehat{A} (e^{i\omega})}{\widehat{C} (e^{i\omega})} \right |^2 \left( \left|G(e^{i\omega})- \widehat{G} (e^{i\omega}) \right|^2 \left|S_{yp}(e^{i\omega}) \right|^2 {\Phi}_{r_ur_u} (\omega) +\right.\right. \newline  \left.\left. \cdots  \left|W(e^{i\omega})\frac{\widehat{S}_{yp}(e^{i\omega})}{S_{yp}^{-1}(e^{i\omega})}-\frac{\widehat{C}(e^{i\omega})}{\widehat{A}(e^{i\omega})} \right|^2  {\Phi}_{ee}(\omega) \right)\right\} \mathrm d\omega $ }\\ \cline{2-3}

& OUTPUT ERROR &{ $ Argmin \displaystyle\int\nolimits_{-\pi}^{+\pi}  \left| {\widehat{S}}_{yp} (e^{i\omega})\right |^2  \left|G(e^{i\omega})- \widehat{G} (e^{i\omega}) \right|^2 \left|S_{yp}(e^{i\omega}) \right|^2 {\Phi}_{r_ur_u} (\omega) \mathrm{d} \omega$}\\

\hline

\end{tabular}

\caption{Bias distribution for PEM algorithms in open- and closed-loop}

\label{table3b}

\end{minipage}}
\end{table}
\end{center}

\section{Bias distribution for PLR algorithms with regressor filtering}

\subsection{Equivalent prediction error expressions}

In classical PLR algorithms, such as those presented in section \ref{sect_bias_PLR}, the parameter adaptation algorithm (PAA) is fed with the  \emph{a posteriori} prediction error $\varepsilon (t+1)$ and the regressor $\phi (t)$. In some algorithms, a filtering is operated on the regressor such that:
\begin{equation} 
\phi _{f}(t)=\frac{1}{Q_{f}(q^{-1})}\phi (t)
\label{e_12}
\end{equation}
 and $\phi _{f}(t)$ is used in the PAA instead of $\phi (t)$. $Q_{f}(q^{-1})$ is a ratio of monic polynomials, and we assume that $Q_{f}(q^{-1})$ does not depend on $\widehat{\theta }$. The purpose of this filtering is usually to relax the convergence conditions, see for example (\cite{r4} p. 161). With this filtering the PAA becomes (the predicted output expression is not modified): 
\begin{subequations}
\begin{align}
\widehat{\theta }(t+1)& =\widehat{\theta }(t)+F(t)\phi _{f}(t)\varepsilon
(t+1) \\
F^{-1}(t+1)& =\lambda _{1}F^{-1}(t)+\lambda _{2}\phi _{f}(t)\phi _{f}^{T}(t)
\label{e_12b}
\end{align}

\begin{theorem}

	In the case of regressor filtering in the PAA, corresponding to equations \eqref{e_12}, \eqref{e_12b}, let $Q (q^{ -1} ,\widehat{\theta }) =1 +\widehat{\theta }^{T} \frac{ \partial \phi}{ \partial (q \varepsilon)}$, \textit{, and assume } $\frac{ \partial ^{2}\phi}{ \partial (q\varepsilon)  \partial \widehat{\theta }} =0 .$ The equivalent prediction error $\varepsilon _{E}(t +1) $, as defined in proposition 2, is given by: 

\begin{itemize}
 \item for the equation error model: \\
					$\varepsilon _{E}(t +1) =\frac{Q(q^{ -1} ,\widehat{\theta })}{Q_{f}(q^{ -1} ,\widehat{\theta })}\varepsilon (t +1)+\left(1 -\frac{Q(q^{ -1} ,\widehat{\theta })}{Q_{f}(q^{ -1}	,	  \widehat{\theta })}\right)e(t +1) $
			
\item for the open-loop output error model: \\
			$\varepsilon _{E}(t +1) =\frac{Q(q^{ -1} ,\widehat{\theta })}{Q_{f}(q^{ -1} ,\widehat{\theta })}\varepsilon (t +1)+\left(1 -\frac{Q(q^{ -1} ,\widehat{\theta })}{Q_{f}(q^{ -1}	,	  \widehat{\theta })} \right)v(t +1) $

	\item for the closed-loop output error model: \\
				$\varepsilon _{E}(t +1) =\frac{Q(q^{ -1} ,\widehat{\theta })}{Q_{f}(q^{ -1} ,\widehat{\theta })}\varepsilon (t +1)+\left(1 -\frac{Q(q^{ -1} ,\widehat{\theta })}{Q_{f}(q^{ -1}	,	  \widehat{\theta })}\right)S_{yp}(q^{-1})v(t +1)$

\end{itemize}

\end{theorem}

\begin{proof}
We give the proof for the open-loop equation error model; it is very similar
to the proof of theorem \ref{theorem1}. Once again we check that if the
system is in the model set and $\theta =\widehat{\theta }$ one has $%
\varepsilon _{E}(t+1)=e(t+1),$ since in this case $\varepsilon (t+1)=e(t+1)$%
. One must verify that $\phi _{E}(t)=\frac{-\partial \varepsilon _{E}(t+1)}{%
\partial \widehat{\theta }}$ is not a function of $\widehat{\theta }$, i.e. $%
\frac{\partial \phi _{E}(t)}{\partial \widehat{\theta }}=0.$ One has 
\begin{eqnarray*}
\frac{\partial \varepsilon _{E}(t+1)}{\partial \widehat{\theta }} &=&\frac{%
Q(q^{-1})}{Q_{f}(q^{-1})}\frac{\partial \varepsilon (t+1)}{\partial \widehat{%
\theta }}+\frac{1}{Q_{f}(q^{-1})}\frac{\partial Q}{\partial \widehat{\theta }%
}\left( \varepsilon (t+1)-e(t+1)\right)  \\
&=&-\phi _{f}(t)+\frac{1}{Q_{f}(q^{-1})}\frac{\partial Q}{\partial \widehat{%
\theta }}(\varepsilon (t+1)-e(t+1)).
\end{eqnarray*}
\newline
But $\frac{1}{Q_{f}(q^{-1})}\frac{\partial Q}{\partial \widehat{\theta }}=%
\frac{\partial \phi _{f}}{\partial (q\varepsilon )}.$ Therefore, 
\begin{equation*}
\frac{\partial \varepsilon _{E}(t+1)}{\partial \widehat{\theta }}=-\phi
_{f}(t)+\frac{\partial \phi _{f}}{\partial (q\varepsilon )}(\varepsilon
(t+1)-e(t+1)).
\end{equation*}
With $\phi _{E}(t)=-\frac{\partial \varepsilon _{E}(t+1)}{\partial \widehat{%
\theta }}$, one obtains again $\frac{\partial \phi _{E}(t)}{\partial 
\widehat{\theta }}=-\frac{\partial \phi _{f}(t)}{\partial \widehat{\theta }}+%
\frac{\partial \phi _{f}(t)}{\partial \widehat{\theta }}=0,$ assuming that $%
\frac{\partial ^{2}\phi }{\partial (q\varepsilon )\partial \widehat{\theta }}%
=0$, i.e. $\frac{\partial ^{2}\phi _{f}}{\partial (q\varepsilon )\partial 
\widehat{\theta }}=0$ (as $Q_{f}(q^{-1})$ does not depends on $\widehat{%
\theta }$). The generalization to open- and closed-loop output error schemes
is immediate by substituting $e(t)$ by $v(t)$ and $S_{yp}(q^{-1})v(t)$,
respectively.
\end{proof}

\noindent This regressor filtering has been applied in the litterature to the output
error algorithms. In open loop, if $Q_{f}(q^{-1})$ is a monic
polynomial with $Q_{f}(q^{-1})=A_{o}(q^{-1})$, the algorithm corresponds to
the so-called F-OLOE \cite{r4}, p.173 of Landau \emph{et al.} In closed-loop
operations, if we use $Q_{f}(q^{-1})=\frac{P_{o}(q^{-1})}{S(q^{-1})}$, we
get Landau's F-CLOE algorithm where $P_{o}$ an estimation of the closed loop
characteristic polynomial. The equivalent prediction error is provided in
Table \ref{table4}, and Table \ref{table5} gives the associated bias
distributions.

\begin{center}
\begin{table}[h]
\resizebox{1\textwidth}{!}{\begin{minipage}{\textwidth}
\renewcommand*{\arraystretch}{1}

\begin{tabular}{|>{\centering\scriptsize\arraybackslash}p{25pt}|>{\centering\scriptsize\arraybackslash}p{70pt}|>{\centering\scriptsize\arraybackslash}p{40pt}|>{\centering\scriptsize\arraybackslash}p{173pt}|}

\hline

 \multicolumn{2}  {|>{\centering\scriptsize}p{85pt}|}{NOISE MODEL} & {\scriptsize $\frac{{Q}{ (\widehat{\theta })}}{Q_f (\widehat{\theta })}$} & {\scriptsize $\varepsilon _{E} (t)$}  \\
\hline  
Open-loop & OUTPUT ERROR & $\frac{\widehat{A}}{A_0}$ & $\frac{\widehat{A}}{A_{0}} (G -\widehat{G}) u (t) +v (t)$  \\
\hline
Closed-loop & OUTPUT ERROR & $\frac{\widehat{P}}{P_0}$ & $\frac{\widehat{A}S}{P_{0}} (G -\widehat{G}) S_{yp}r_{u} (t) +S_{yp}v(t)$\\
\hline

\end{tabular}

\caption{Equivalent prediction error in the case of regressor filtering}

\label{table4}

\end{minipage}}
\end{table}

\begin{table}[]
\resizebox{1\textwidth}{!}{\begin{minipage}{\textwidth}
\renewcommand*{\arraystretch}{1}

\begin{tabular}{|>{\centering\scriptsize\arraybackslash}p{25pt}|>{\centering\scriptsize\arraybackslash}p{70pt}|>{\centering\scriptsize\arraybackslash}p{225pt}|}

\hline

 \multicolumn{2}  {|>{\centering\scriptsize}p{85pt}|}{NOISE MODEL} & {${\widehat{\theta}}_{PLR}^*$}  \\
\hline  
Open-loop & OUTPUT ERROR & { $ Argmin \displaystyle\int\nolimits_{-\pi}^{+\pi}  \left| \frac{\widehat{A} (e^{i\omega})}{{A_0} (e^{i\omega})} \right |^2  \left|G(e^{i\omega})- \widehat{G} (e^{i\omega}) \right|^2 {\Phi}_{uu} (\omega)  \mathrm d\omega $ }\\ 
\hline
Closed-loop & OUTPUT ERROR & { $ Argmin \displaystyle\int\nolimits_{-\pi}^{+\pi}  \left| \frac{\widehat{A} (e^{i\omega})S(e^{i\omega})}{{P}_0 (e^{i\omega})} \right |^2  \left|G(e^{i\omega})- \widehat{G} (e^{i\omega}) \right|^2 \times \newline\cdots\left|S_{yp}(e^{i\omega}) \right|^2{\Phi}_{r_ur_u} (\omega)  \mathrm d\omega $ }\\ 
\hline

\end{tabular}

\caption{Bias distribution in case of regressor filtering in the parameter adaptation algorithm}

\label{table5}

\end{minipage}}
\end{table}
\end{center}

\noindent In \cite{r4} p.175 and p. 300,  adaptive versions of F-OLOE and F-CLOE algorithms are
proposed, consisting in imposing respectively: $A_{0}(q^{-1},t)=\widehat{A}%
(q^{-1},t)$ and $P_{0}(q^{-1},t)=\widehat{A}(q^{-1},t)S(q^{-1})+\widehat{B}%
(q^{-1},t)R(q^{-1})$. We notice that, according to Table \ref{table5}, the
associated bias distribution corresponds to that of PEM algorithms in Table %
\ref{table3b}. This confirms that the so-called AF-OLOE and AF-CLOE algorithms are recursive PEM algorithms, a point which was already suggested in \cite{r4}, p. 308.

\section{Simulations}

In order to show the relevance of the above analysis, we propose a set of
simulations based on the same example as the one presented in (\cite{r3},
section 6). It concerns the closed-loop identification of a system $G(q^{-1})
$: 
\end{subequations}
\begin{equation}
G(q^{-1})=\frac{q^{-1}+0.5q^{-2}}{(1-1.5q^{-1}+0.7q^{-2})(1-q^{-1})}  \notag
\end{equation}

\noindent fed back by a controller $K(q^{-1})$: 
\begin{equation}
K(q^{-1})=\frac{0.8659-1.2763q^{-1}+0.5204q^{-2}}{(1-q^{-1})(1+0.3717q^{-1})}
\notag
\end{equation}

\noindent The excitation is a 9 registers PRBS introduced additively on the
system input. The internal PRBS period being equal to the sampling period,
the spectral power density can be considered constant over the frequency
range. The model order is chosen equal to 2, as in \cite{r3}, thus a bias is
necessarily present on the identified model. We compare the identification
results of two models:

\begin{itemize}
\item the recursive PLR CLOE algorithm without prediction error filtering,
method as employed in \cite{r3},

\item a non recursive PEM CLOE method without prediction error filtering, by
computing the sensitivity function of the prediction error with respect to
the estimated parameters, and performing a non-linear optimization.
\end{itemize}

\noindent We recall the theoretical bias distribution of CLOE scheme without
regressor filtering derived in this paper for PLR methods (see the last row
of Table 3, respectively):

\begin{equation}
\widehat{\theta }_{PLR}^{\ast }=Arg\min \int\nolimits_{-\pi }^{+\pi
}\left\vert \widehat{A}(e^{i\omega })\right\vert ^{2}\left\vert G(e^{i\omega
})-\widehat{G}(e^{i\omega })\right\vert ^{2}\left\vert S_{yp}(e^{i\omega
})\right\vert ^{2}\Phi _{r_{u}r_{u}}(\omega )\mathrm{d}\omega   \label{eq901}
\end{equation}

\noindent and for PEM  (last row of Table 4):

\begin{equation}
\widehat{\theta }_{PEM}^{\ast }=Arg\min \int\nolimits_{-\pi }^{+\pi
}\left\vert \frac{\widehat{A}(e^{i\omega })\widehat{S}(e^{i\omega })}{%
\widehat{P}(e^{i\omega })}\right\vert ^{2}\left\vert G(e^{i\omega })-%
\widehat{G}(e^{i\omega })\right\vert ^{2}\left\vert S_{yp}(e^{i\omega
})\right\vert ^{2}\Phi _{r_{u}r_{u}}(\omega )\mathrm{d}\omega   \label{eq902}
\end{equation}

\noindent where $\widehat{P}(q^{-1})=\widehat{A}(q^{-1})S(q^{-1})+\widehat{B}%
(q^{-1})R(q^{-1})$. \newline
\noindent We notice that these two expressions differ by the
introduction of the term $\left|\frac{ {S} (e^{i\omega})}{\widehat{P}
(e^{i\omega})}\right |^2$  in the PEM scheme. Therefore we can expect that
the PEM model fit must be better than the PLR model fit at frequencies where
the filter $\frac{S(q^{-1})}{\widehat{P}(q^{-1})}$ magnitude is the largest,
and worse at frequencies where this magnitude is small. This is exactly what is
observed. Figure \ref{fig1} displays the Bode diagram of the real system compared to those of the two
identified models. Figure \ref{fig2} top graph shows the
Vinnicombe gap (see \cite{r12}) from the real system to these two
identified models, whereas the bottom graph displays the magnitude of the
filter $\frac{S(q^{-1}}{\widehat{P}(q^{-1}}$. PEM effectively
provides a better adjustment around normalized frequencies equal to 0.15,
corresponding to the largest magnitude of the said filter. Concomitantly, we
notice a significant misfit at high frequency (compared to the PLR method),
and another misfit at very low frequency. The misfit at low frequency is not
as important as at high frequency. It can be explained by the influence of
the second weighting function $ {\left| S_{yp}(e^{i\omega}) \right |}^2 $
entering in the expressions \eqref{eq901}, \eqref{eq902}, that has a decreasing
magnitude towards zero, because of an integrator in the controller.

\begin{figure}[H]
\begin{center}
\includegraphics[ width=3.2in, keepaspectratio]{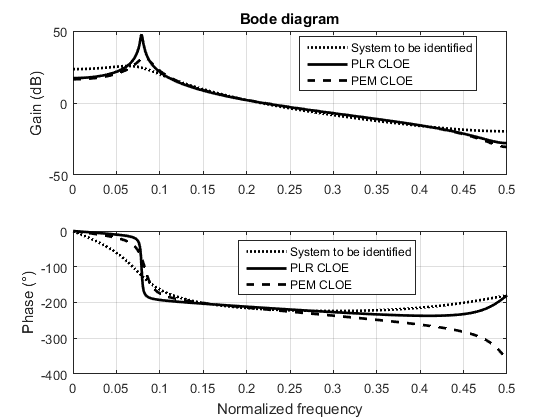}
\end{center}
\caption{ PLR CLOE and PEM CLOE models compared to the real system}
\label{fig1}
\end{figure}

\begin{figure}[H]
\begin{center}
\includegraphics[ width=3.2in, keepaspectratio]{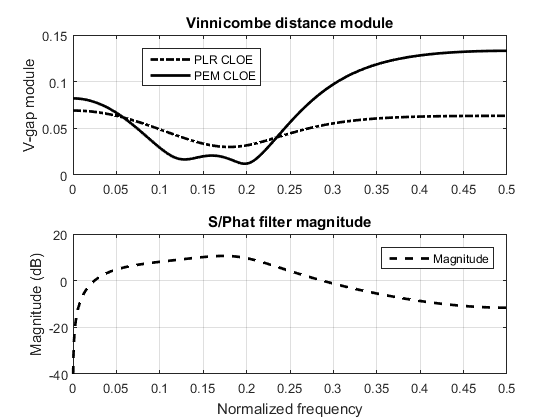}
\end{center}
\caption{ Vinnicombe gap from PLR CLOE and PEM CLOE models to the real
system, and filter $\frac{S(q^{-1})}{\protect\widehat{P}(q^{-1})}$
magnitude. }
\label{fig2}
\end{figure}

\section{Conclusion}

We have demonstrated that the bias distribution expression in the frequency
domain of pseudo-linear regression identification algorithms differs from
prediction error method schemes, both in open- and closed-loop operations.
In particular, we have shown that pseudo-linear regression algorithms do not
minimize the prediction error variance. That led us to introduce the concept
of equivalent prediction error whose variance is effectively minimized by
the pseudo-linear regression schemes. We have derived a list of bias
distributions for algorithms belonging to this class. These expressions
show that, compared to prediction error methods, most of the pseudo-linear
regression identification structures --without regressor filtering-- provide a
better model fit in high frequency. This filtering in the parameter
adaptation algorithm induces a change in the bias distribution, therefore it
can be used not only to relax the convergence conditions of algorithms, but
also to enhance or depress the model accuracy over a given frequency range.

\appendix
\section*{Appendix}
\label{appendix_A} Extended Hilbert spaces are Fr{\'e}chet spaces, not
Banach spaces, so that the classical differential calculus does not apply in
this context. Some hints are given here on the differential calculus in
extended Hilbert spaces. For the sake of simplicity, the Hilbert space
considered is the space $l_2$ of square-summable sequences which are zero
for $t<0$. The truncation operator $P_{t}$, determined for all $t\geq0$ and
any $x\in l_2$ by $P_{t}x(\tau)=x(\tau)$ if $t\leq \tau$, and $P_{t}x(\tau)=0$
otherwise, is an orthonormal projection of $l_2$, and a resolution of
identity. An operator $\mathcal{G}: l_{2e}\rightarrow l_ {2e} $ is causal if
and only if $\mathcal{P}_{t}\circ \mathcal{G} \circ \mathcal{P}_{t}=\mathcal{%
P}_{t} \circ \mathcal{G} $ for all $t\in\mathbb{Z}_{+} $ (see \cite{r9},
section 1.B). 
As easily shown, a causal operator is continuous at $u^{0}\in l_{2e}$ (in
the usual sense, when the domain and the codomain of $G$ are both endowed
with their canonical locally convex topology) if, and only if $P_{t}\circ
G\circ P_{t}$, viewed as an operator from $l_{2}$ to $l_{2},$ is continuous
at $u_{t}^{0}$ for all $t\in \mathbb{Z}_{+}.$ \ Every causal linear operator
is continuous, as shown by the proof of (\cite{r10}, Section 7.2, Thm. 22)
(mutatis mutandis).

A causal operator is called differentiable at $u^{0}_t \in l_{2e} $ if $%
\mathcal{P}_{t}\circ \mathcal{G} \circ \mathcal{P}_{t} $ is differentiable
at $u^{0}_t $, for all $t \in \mathbb{Z}_{+} $. Its differential is then
defined to be the mapping $\mathcal{DG}(u^{0}) \in \mathcal{L}%
_{c}(l_{2e},l_{2e}) $, where $\mathcal{L}_{c}(l_{2e},l_{2e}) $ is the space
of causal linear mappings from $l_{2e} $ to $l_{2e}$, uniquely determined by
the condition: $\mathcal{D}(\mathcal{P}_{t}\circ \mathcal{G} \circ \mathcal{P%
}_{t})=\mathcal{P}_{t}\circ \mathcal{DG} \circ \mathcal{P}_{t}$ for all $t
\in \mathbb{Z}_{+} $. As a result, if $\mathcal{G} $ is causal and linear,
it is differentiable at any point $u^{0}\in l_{2e}$, and $\mathcal{DG}(u^0)=%
\mathcal{G} $. 
Let $x $ and $w $ be two causal signals belonging to $l_{2e}$, and $\mathcal{%
G}(q^{-1}) $ be a causal operator. If one has : $w\text{~}=\text{~}\mathcal{G}%
(q^{-1})x$, then the following relation holds for any $x_{0}\in l_{2e}$: 
\begin{equation*}
\frac{\partial{w}}{\partial{x}}\left(x_{0}\right)=\mathcal{G}(q^{-1})
\end{equation*}


\end{document}